\documentclass[epj]{svjour} 

\usepackage[utf8]{inputenc} 

\usepackage{graphicx,bbm,amsfonts,amssymb,amsopn,amsmath}
\usepackage{dsfont}
\usepackage{calligra}

\usepackage{hyperref}

\def\be{\begin{equation}}
\def\ee{\end{equation}}
\def\bra#1{\mathinner{\langle{#1}|}}
\def\ket#1{\mathinner{|{#1}\rangle}}

\def\aa{{\rm a}}

\newcommand{\ave}[1]{{\langle #1\rangle}}

\newcommand{\ua}{\uparrow}
\newcommand{\da}{\downarrow}
\newcommand{\ii}{ {\rm i} }
\newcommand{\dd}{ {\rm d} }

\newcommand{\ZZ}{\mathbb{Z}}

\newcommand{\CC}{\mathbb{C}}

\newcommand{\z}{{\rm z}}

\newcommand{\hm}[2]{{\hat{#1}^{(#2)}}}
\newcommand{\hmp}[2]{{\hat{#1}^{'(#2)}}}
\newcommand{\hmu}[2]{{\hat{\mathbf{#1}}_{#2}}}
\newcommand{\hmup}[2]{{\hat{\mathbf{#1}}'_{#2}}}
\newcommand{\mmu}[2]{{\mathbf{#1}_{#2}}}
\newcommand{\mmup}[2]{{\mathbf{#1}'_{#2}}}

\def\WW{{{\hat{\mathbf{W}}}}}

\newcommand{\LL}{{\hat{\cal L}}}

\newcommand{\half}{{\textstyle\frac{1}{2}}}

\newcommand{\mm}[1]{{\mathbf{#1}}}
\newcommand{\bb}[1]{{\mathbf{#1}}}

\def\tr{{{\rm tr}}}

\def\one{\mathbbm{1}}

\newcommand{\RR}{\mathbb{R}}
\newcommand{\NN}{\mathbb{N}}

\begin{document}

\title{Strongly correlated non-equilibrium steady states with currents ---  quantum and classical picture}

\author{Berislav Bu\v{c}a\inst{1,2} and Toma\v{z} Prosen\inst{3} }
\titlerunning{Strongly correlated non-equilibrium steady states \ldots}
\institute{Department of Medical Physics and Biophysics, University of Split School of Medicine, 21000 Split, Croatia \and Clarendon Laboratory, University of Oxford, Parks Road, Oxford OX1 3PU, United Kingdom \and Department of physics, Faculty of Mathematics and Physics, University of Ljubljana,
Jadranska 19, SI-1000 Ljubljana, Slovenia}
\abstract{In this minireview we will discuss recent progress in the analytical study of current-carrying non-equilibrium steady states (NESS) that can be constructed in terms of a matrix product ansatz.  We will focus on one-dimensional exactly solvable strongly correlated cases, and will study both quantum models, and classical models which are deterministic in the bulk. The only source of classical stochasticity in the time-evolution will come from the boundaries of the system. Physically, these boundaries may be understood as Markovian baths, which drive the current through the system. The examples studied include the open XXZ Heisenberg spin chain, the open Hubbard model, and a classical integrable reversible cellular automaton, namely the Rule 54 of Bobenko {\em et al}. [Commun. Math. Phys. {\bf 158}, 127 (1993)] with stochastic boundaries. The quantum NESS can be at least partially understood through the Yang-Baxter integrability structure of the underlying integrable bulk Hamiltonian, whereas  for the Rule 54 model NESS seems to come from a seemingly unrelated integrability theory. In both the quantum and the classical case, the underlying matrix product ansatz defining the NESS also allows for construction of novel conservation laws of the bulk models themselves. In the classical case, a modification of the matrix product ansatz also allows for construction of states beyond the steady state (i.e., some of the decay modes -- Liouvillian eigenvectors of the model). We hope that this article will help further the quest to unite different perspectives of integrability of NESS (of both quantum and classical models) into a single unified framework.} 

\maketitle

\section{Introduction}

Studying strongly-correlated (interacting) many-body systems is challenging. Only several classes of models admit analytical treatment, either in quantum \cite{bethe1,bethe2,bethe3,sutherland} or classical \cite{fadtak} Hamiltonian context. Such models (in particular in the quantum domain) typically poses families of local and quasi-local conservation laws \cite{Prosenreview,Prosen1,ProsenIlievski,IMP,IMPZ} (or the corresponding particle content \cite{Enej1}) that have important ramifications on their physical behavior, especially transport \cite{Prosen1} and relaxation towards equilibrium \cite{revieweq}. These models fall under the category of Bethe ansatz, or Yang-Baxter, solvability \cite{bethe1,bethe2,bethe3} and are called integrable. Important physically relevant quantum integrable models include the Heisenberg XXZ spin chain \cite{bethe2} and the one-dimensional Hubbard model \cite{Hubbardbook}. In algebraic analogy, one can define integrability also for classical stochastic many-body systems, the prominent examples of which
include the various exclusion process models, such as the partially asymmetric simple exclusion process \cite{ASEPreview}. The link between the classical stochastic models and the quantum models is achieved by performing a simple Wick rotation. 

Beyond the Bethe ansatz (or Yang-Baxter) paradigm little is known. Generalizing the Bethe ansatz equations to more than one (plus one) dimension has been attempted, though the solutions found appear to be very special and fine tuned \cite{tetrahedron}. However, a possibility of generalizing the notion of exact solvability beyond the Bethe ansatz para\-digm remains a completely open issue. 

Apart from the aforementioned issues of conservation laws of integrable models, another key aspect of interest in statistical physics out-of-equilibrium is the nature and relaxation towards a non-equilibrium steady state (NESS) of a model connected to baths. For conceptual and mathematical simplicity much recent research focuses on modeling the baths as Markovian reservoirs, both for quantum models, e.g, Refs.~\cite{Arenas1,Ilievski,EverestLesanovsky,ZnidScarVar,ManzanoChuangCao,KPS2,IlievskiZunkovic,Prosen2,WolffSheikhanKollath} (for a recent review in the context of exact solutions see \cite{Prosenreview}), classical stochastic models \cite{ASEPreview} and classical deterministic models  \cite{ProsenMejia,ProsenBuca} driven by simple stochastic reservoirs
that act locally on the boundaries of the model.

In this minireview we will discuss the common aspects of both quantum and classical deterministic models driven by Markovian baths acting on the boundaries. The examples we will look at include the open XXZ model (the NESS of which was first solved in \cite{Prosen1,Prosen2}), the open fermi Hubbard model \cite{ProsenHubbard,PopkovProsen} and the classical deterministic bo\-undary driven Rule 54 cellular automaton (the NESS was constructed in \cite{ProsenMejia} and the dominant decay modes were solved in \cite{ProsenBuca}).

The basic conceptual common aspects of all the studied models are that the models
\begin{itemize}
\item are all deterministic and hence time-reversible in the bulk (i.e., there is no stochasticity),
\item all admit solutions for the NESS in terms of a matrix product ansatz (MPA), and
\item all the solutions for the NESS may be used to construct novel kinds of conservation laws for the bulk models, unobtainable by standard techniques.
\end{itemize}
The main conceptual differences are that only the NESS of the quantum models can be understood in terms of Bethe ansatz integrability. The reversible cellular automaton Rule 54 seems to be a part of a different class of exactly solvable integrable models. The physical properties of the models, including the nature of the transport in the steady state, are quite diverse. 

 This paper is organized as follows. In Sect.~\ref{sec2} we will discuss the quantum models. We will first describe the general setup of Markovian baths modeled with the Lindblad master equation \cite{Lindblad,GKS76,openbook}. We then continue by discussing briefly algebraic Bethe ansatz techniques and how they relate to the quantum integrability of the NESS. Two specific examples, the XXZ spin chain and the Hubbard model are discussed. The solution for the the Hubbard model necessitates the introduction of a novel, infinitely dimensional Lax operator which has not been used in other (equlibrium) contexts. This section will follow the results given in \cite{Prosenreview,Prosen1,Prosen2,PopkovProsen}. In Sect.~\ref{sec3} we begin by introducing the concept of classical deterministic, reversible cellular automata on a diamond shaped lattice in 1+1 dimensions, and continue by discussing the notion of \emph{holographic ergodicity}, which is a 
 condition on the bulk deterministic dynamics and boundary stochastic processes in the Markov matrix propagating the time evolution, which guarantees that the boundary driven system will relax to a unique steady state (NESS). We then show how to calculate the NESS of a boundary driven Rule 54 cellular automaton by introducing a novel cancellation mechanism algebra used to construct the MPA. We later generalize this to obtain the most important part of the Liouvillian decay spectrum (including the leading decay mode). Sect.~\ref{sec3} mostly follow Refs. \cite{ProsenMejia,ProsenBuca}.
 We conclude by discussing some notable open questions in Sect.~4. 

\section{Quantum integrable non-equilibrium steady states}
\label{sec2}

One of the most analytically tractable methods of studying interacting many-body quantum systems out-of-equili\-brium is the framework based on the Lindblad master equation \cite{Lindblad,GKS76}, and it has attracted a lot of attention recently, e.g. see \cite{Ilievski,KPS2,ManzanoHurtado,Garrahan,ZnidaricFCS,Monthus,PigeonXuereb}. In this setup the system is driven out-of-equilibrium by a set of baths and is modelled by dividing the generator of the time evolution into two parts: 1) the Hamiltonian $H$, representing the physical system itself, and 2) a set of Lindblad jump operators $L_\mu$, representing the action of the baths on the system. The whole time evolution is given by the following 
Lindblad-Gorini-Kosakowski-Sudarshan master equation,
\begin{eqnarray}
&&\frac{\dd}{\dd t}\rho(t) = \LL\rho(t):= \nonumber \\
&&-\ii [H,\rho(t)] + \sum_\mu \left(2L_\mu \rho(t) L_\mu^\dagger - \{L_\mu^\dagger L_\mu,\rho(t)\}\right),
\label{eq:lindeq}
\end{eqnarray}
where, 
\begin{itemize}
\item $\rho(t)$ is the density matrix of the system at time $t$, and
\item $\LL$ is the generator of the full time evolution. It is a superoperator acting on the vector space of linear operators  ${\rm End}({\cal H})$ over the Hilbert space of the system ${\cal H}$.
\end{itemize}
The equation describes an open system which can, depending on the Lindblad jump operators $L_\mu$, exchange energy, particles, etc. with the environment (the baths). 
In the long-time limit ($t \to \infty$) for finite systems the time-evolution relaxes the system to a non-equlibrium steady state (NESS) $\rho_\infty$ \cite{Spohnreview}. The equation for the NESS is given as,
\begin{equation}
\frac{\dd}{\dd t}\rho_\infty= \LL\rho_\infty=0, \label{NESS}
\end{equation}
and is thus an eigenvector of $\LL$ of eigenvalue zero, $\rho_\infty \in {\rm ker} \LL$.
This steady state may be degenerate \cite{sym1,sym2,sym3} or unique \cite{Spohnreview,Evans} depending on whether the generator $\LL$ fulfils different theorems that we do not discuss here (see e.g. \cite{Spohnreview}).  

In certain cases it is possible to find the NESS analytically. Various cases when the many-body system $H$ is non-interacting or equivalent to a non-interacting system are known (e.g., \cite{P08,Z10,Z11,Kos}), but we will be interested here in the cases when $H$ is genuinely interacting and integrable \cite{Prosen1,Ilievski,IlievskiZunkovic,Prosen2,ProsenHubbard,PopkovProsen,KPS,IlievskiProsenLai} and the baths act only on the boundaries of the system (\emph{boundary-driving}).  More specifically, we will focus on two illustrative examples: the open boundary driven Heisenberg XXZ spin $1/2$ chain and the boundary driven open 1D fermionic Hubbard model. It is important to note that th the bulk system $H$ is purely Hermitian and the Lindblad jump operators $L_\mu$, which we will study, are purely incoherent and thus represent the action of classically stochastic baths on the system. 

\subsection{The maximally boundary-driven XXZ spin chain}
\label{sec:XXZ}
The NESS of the open XXZ spin chain was first solved in the pair of Letters \cite{Prosen1,Prosen2} using a novel kind of 
tri-diagonal infinite-dimensional matrix product ansatz. This solution was later understood in terms of the underlying Bethe ansatz integrability of the XXZ model \cite{Prosenreview,PIP,KPS,IlievskiPhD}. 

We begin by defining the model. The Hamiltonian of the XXZ spin chain of $n$ spin sites is given as,
\begin{eqnarray}
H^{\rm{XXZ}}&=&\sum_{j=1}^{n}h^{\rm{XXZ}}_{j,j+1}= \nonumber \\
&=&\sum_{j=1}^{n}2(\sigma^{+}_{j}\sigma^{-}_{j+1}+\sigma^{-}_{j}\sigma^{+}_{j+1})+
\Delta \sigma^{\rm{z}}_{j}\sigma^{\rm{z}}_{j+1}. \label{ham}
\end{eqnarray}
The operators $\sigma^{\alpha}_{j}$ are the standard Pauli matrices acting on site $j$. More specifically,
\begin{equation}
\sigma^{\alpha}_{j}\equiv  \one_{2^{j-1}}\otimes \sigma^{\alpha}\otimes \one_{2^{n-j}}.
\label{eqn:spin_varibles}
\end{equation}
The parameter $\Delta$ is the anisotropy and depending on the value of $\Delta$ the model exhibits quite different physical properties. At the isotropic point $\Delta=1$ there is a quantum phase transition from a gapless ($|\Delta|<1$) to a gapped ($|\Delta|>1$) phase.

The Lindblad jump operators are given as, 
\begin{equation}
L_{1}=\sqrt{\varepsilon}\sigma^{+}_{1},\qquad L_{2}=\sqrt{\varepsilon}\sigma^{-}_{n},
\label{maxdriv}
\end{equation}
where $\varepsilon$ is the system-bath coupling and represents the strength with which the baths act on the system. It is also equal to the stochastic rate with which the baths act on the system \cite{openbook}. The system is considered maximally driven as the bath coupled to the first site 1 only, injects spin into (and the bath coupled to the last site $n$ only ejects spin out of) the system. In other words, we have an ideal incoherent source of magnetization attached to the left end and ideal sink of magnetization attached to the right end. Thus, the time evolution in the bulk is fully coherent and the baths act fully incoherently only on the boundary sites. 

To present the solution to Eq.~\ref{NESS} and find the NESS of the system we will first begin by introducing a few important concepts in quantum integrability. For an actual introduction to the topic of algebraic Bethe ansatz and quantum integrability the reader is referred to Refs. \cite{bethe1,bethe2,bethe3}. 

 Let us first introduce an \emph{auxiliary} space ${\cal H}_\aa$ that may be finite or infinite-dimensional. Operators acting on this space will be denoted with bold letters. Operators acting on the physical space ${\cal H}_p$ of the system in addition to this auxiliary space will have a number subscript denoting on which spin site of the system they act. 

We begin by studying a particular form of Yang-Baxter equation, the so-called RLL relation, 
\begin{equation}
R_{1,2}(\varphi_1-\varphi_2)\bb{L}_{1}(\varphi_1)\bb{L}_{2}(\varphi_2) = 
\bb{L}_{1}(\varphi_2)\bb{L}_{2}(\varphi_1)R_{1,2}(\varphi_1-\varphi_2),
\label{RLL}
\end{equation}
where $R_{1,2}$ is called the \emph{$R$-matrix} and $\bb{L}_j$ are called the \emph{Lax operators} and the parameters $\varphi_j$ are called the \emph{spectral parameters}. The quantum integrability of the XXZ spin chain is based on the underlying quantum group symmetry of the model, namely the $\calligra{U_{q}}  ({\mathfrak{sl}}_{2}) $ algebra, which is a quantum deformation of the ${\frak{sl}}_{2}$ algebra, or angular momentum algebra, and is defined by the standard $q-$deformed algebraic relations of its generators, ${\bb{s}}^{\rm{z}}_{s}, {\bb{s}}^{+}_{s}, {\bb{s}}^{-}_{s}$,
 \begin{equation}
[{{\bb{s}}}^{+}_{s},{{\bb{s}}}^{-}_{s}]=[2 {{\bb{s}}}^{\rm{z}}_{s}]_{q},\quad [{\bb{s}}^{\rm{z}}_{s},{\bb{s}}^{\pm}_{s}]=\pm {\bb{s}}^{\pm}_{s}, \label{algebra}
\end{equation}
where $[x]_q := (q^x-q^{-x})/(q-q^{-1})$ and $q$ is the deformation parameter related to the anisotropy $\Delta$ as $\Delta= \rm{cos} (\gamma)$, $q=\rm{exp} (\ii \gamma)$. A (generically) infinite-dimensional highest weight representation of this algebra over the auxiliary space ${\cal H}_a$ is given as,
\begin{eqnarray}
{{\bb{s}}}^{\rm{z}}_{s}&=&\sum_{k=0}^{\infty}(s-k)\ket{k}\bra{k},\nonumber \\
{{\bb{s}}}^{+}_{s}&=&\sum_{k=0}^{\infty}[k+1]_{q}\ket{k}\bra{k+1},\nonumber \\
{{\bb{s}}}^{-}_{s}&=&\sum_{k=0}^{\infty}[2s-k]_{q}\ket{k+1}\bra{k},
\end{eqnarray}
and is characterised by a continuous complex spin parameter $s$. Note that for certain values of $s$ the representation reduces to a finite-dimensional one. For instance, for half-integer $s$, the above sums truncate due to existence of lowest weight states and ${{\bb{s}}}^{j}_{s}$
generate the standard unitary irreducible representations.

The $4\times 4$ $R$-matrix of the XXZ spin chain reads,
\begin{eqnarray}
R_{1,2}(\varphi)&=& \frac{\sin\varphi}{2}(h^{\rm{XXZ}}_{1,2} + \one\cos\gamma)  \\
&-&\frac{1+\cos\varphi}{2} \one\sin\gamma + \frac{1-\cos\varphi}{2}\sigma^\z_1 \sigma^\z_2 \sin\gamma  \nonumber.
\label{RXXZ}
\end{eqnarray}
The Lax matrix can be given in terms of the generators of $\calligra{U_{q}}  ({\mathfrak{sl}}_{2}) $ as, 
\begin{equation}
{\bb{L}}(\varphi,s)=
\begin{pmatrix}
\sin{(\varphi +\gamma {\bb{s}}^{\rm{z}}_{s})} & (\sin{\gamma}){\bb{s}}^{-}_{s} \\
(\sin{\gamma}){\bb{s}}^{+}_{s} & \sin{(\varphi -\gamma {\bb{s}}^{\rm{z}}_{s})} 
\end{pmatrix}, \label{lax}
\end{equation}
which is given as a matrix on a one-site physical space and we write the dependence on the spin parameter $s$ explicitly. 

By expanding the RLL relation~(\ref{RLL}) with $\varphi_{1,2} = \varphi \pm \delta$ to first order in $\delta$, we find the Sutherland equation \cite{Prosenreview,Sutherland},
\begin{align}
&[h^{\rm{XXZ}}_{1,2},\bb{L}_{1}(\varphi,s)\bb{L}_{2}(\varphi,s)] \label{sutherland}\\
&= -2\sin\gamma \left(\widetilde{\bb{L}}_{1}(\varphi,s)\bb{L}_{2}(\varphi,s) - \bb{L}_{1}(\varphi,s)\widetilde{\bb{L}}_{2}(\varphi,s)\right) \nonumber
\end{align}
where
${\widetilde{\bb{L}}}(\varphi,s) := \partial_{\varphi}{\bb{L}}(\varphi,s)$. 

We can now state the solution to the equation for the NESS \eqref{NESS} in the terms of a matrix product ansatz,
\begin{equation}
\rho_{\infty}=\frac{\Omega \Omega^{\dagger}}{\tr \Omega \Omega^\dagger}, \label{ansatz} 
\end{equation}
with
\begin{equation}
\Omega=\bra{0}\bb{L}_1(\varphi,s)\bb{L}_2(\varphi,s)\ldots \bb{L}_n(\varphi,s)\ket{0}, \label{sol}
\end{equation}
where $\bb{L}_j(\varphi,s)$ is the Lax operator defined by Eq. \eqref{lax} acting on physical spin site $j$ and $\ket{0}$ is the highest weight state in the auxiliary space ${\cal H}_\aa$.
By inserting the ansatz \eqref{ansatz} into the equation for the NESS \eqref{NESS} and using the Sutherland equation \eqref{sutherland}, we are left with a set of equations only on the boundaries (coming from the Lindblad operators $L_1$ and $L_2$) for the spectral parameter $\varphi$ and the spin parameter $s$. These boundary equations may be solved uniquely with,
\begin{align}
\varepsilon&=4\ii[s]^{-1}_{q}\cos{(\gamma s)}, \\
\varphi&=0,
\label{solution}
\end{align}
which yields an imaginary spin parameter $s$ as a function of dissipation rate $\varepsilon$. We note that this setup can be 
generalized to include arbitrary boundary magnetic fields in the Hamiltonian and left-right asymmetric dissipation rates $\varepsilon_{\rm L,R}$ which in turn fix two fully complex parameters $s,\varphi$ \cite{Prosenreview,Chihiro}.

This solution allows us to compute all the observables in the NESS efficiently and, in the thermodynamic limit $n \to \infty$, all the the correlation functions \cite{BucaProsenCC}. More specifically, the spin current is ballistic ($\Delta <1$), sub-diffusive scaling asymptotically as $~1/n^2$ ($\Delta=1$) and exponentially decaying (insulating for $\Delta>1$) \cite{Prosen1} , see Fig.~\ref{fig:XXZ}.

Note as well that, as a consequence of the Yang-Baxter equation, the operator $\Omega(s,\varphi)$ forms a commuting family
of {\em non-diagonalizable} triangular (in the eigenbasis of $\sigma^\z_j$) transfer matrices
\begin{align}
[\Omega(s,\varphi),\Omega(s',\varphi')] = 0,\quad\forall s,s',\varphi,\varphi'\in\CC,
\end{align}
which generate quasilocal conserved charges \cite{Prosen1,ProsenIlievski,IMPZ} going beyond the standard families of strictly local charges generated from a fundamental auxiliary space ($s=1/2$) representation of the transfer matrix \cite{bethe2,bethe3}.

\begin{figure*}
 \centering	
\vspace{-1mm}
\includegraphics[width=1.8\columnwidth]{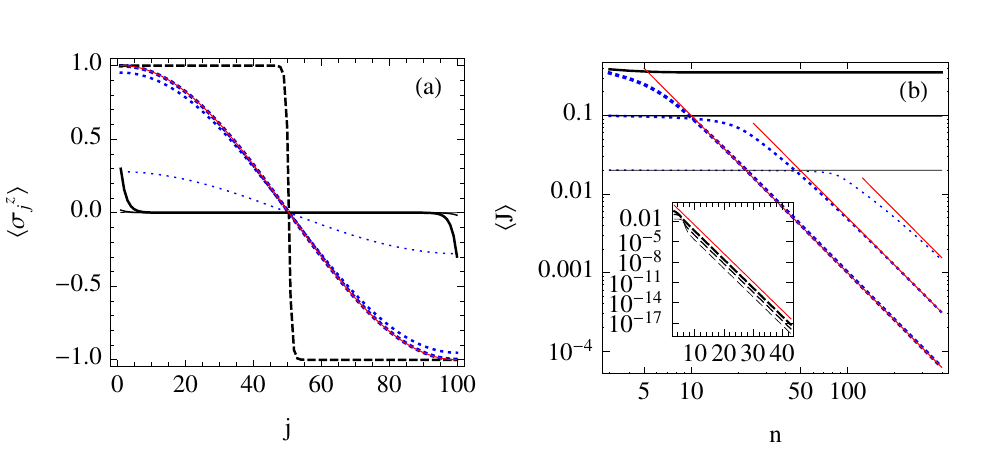}
\vspace{-1mm}
\caption{(From Ref.~\cite{Prosen2}) Spin magnetization $\ave{\sigma^{\rm z}_j}$ profiles for $n=100$ computed from the solution for the NESS (a), and spin currents $\ave{J}$ depending on system size $n$ (b), for boundary driven XXZ model with $\Delta=3/2$ (dashed), $\Delta=1$ (dotted/blue), 
	$\Delta=1/2$ (full curves), shown for three different couplings $\varepsilon=1,1/5, 1/25$ with thick, medium, thin curves, respectively. 
	Red full curves show closed-form asymptotic results (see \cite{Prosenreview} for details) for $\Delta=1$ in the main panels (a,b),
	and  $\Delta>1$ in (b)-inset.}
\label{fig:XXZ}
\end{figure*}

\subsection{The maximally driven open fermionic Hubbard model}

This subsection will present the results of Refs. \cite{ProsenHubbard,PopkovProsen}. The fermionic Hubbard model with $n$ sites is given by the following Hamiltonian,
\begin{align}
&H=-2\sum_{s,x} (c^\dagger_{s,j}c_{s,j+1}+c^\dagger_{s,j+1}c_{s,j} ) \nonumber\\
&+  u \sum_{j} (2n_{\uparrow,j}-1)(2n_{\downarrow,j}-1) \nonumber\\
&+ \mu_{\rm L}(n_{\uparrow,1}+n_{\downarrow,1}-1)+\mu_{\rm R}(n_{\uparrow,n}+n_{\downarrow,n}-1), 
\label{eq:HubHfermi}
\end{align}
where $c_{s,j}$ are the Fermi operators acting on site $j$ and $j\in\{1\ldots n\}$, $s\in\{\ua,\da\}$ denotes each fermion's spin, $u$ is a parameter determining the interaction strength between the fermions on nearest-neighbor sites, $n_{s,j}:=c^\dagger_{s,j} c_{s,j}$ is the local density, and $\mu_{\rm L}$ ($\mu_{\rm R}$) are the left (right) boundary chemical potentials. We take the following Lindblad jump operators,
 \begin{align}
&L_1 = \sqrt{\varepsilon_{\rm L}}c^\dagger_{\uparrow,1},\qquad L_2 = \sqrt{\varepsilon_{\rm L}}c^\dagger_{\downarrow,1}, \nonumber \\ 
&L_3 = \sqrt{\varepsilon_{\rm R}}c_{\uparrow,n},\qquad L_4 = \sqrt{\varepsilon_{\rm R}}c_{\downarrow,n}. \label{lindhub}
\end{align}
They inject (eject) fermions with rates $\varepsilon_{\rm L}$ on site $1$ ($\varepsilon_{\rm R}$ on site $n$), so we have again an ideal incoherent source/sink situation.

Upon performing the Wigner-Jordan transformation, 
\begin{equation}
c_{\uparrow,j}=P^{(\sigma^1)}_{j-1} \sigma_j^{1,-},\quad c_{\downarrow,j}=P^{(\sigma^1)}_n P^{(\sigma^2)}_{j-1} \sigma_j^{2,-}
\end{equation}
where
$P^{(\sigma^1)}_j:=\sigma_1^{1, \rm z}\sigma_2^{1,\z} \cdots \sigma_j^{1, \rm z}$, $P^{(\sigma^2)}_j:=\sigma_1^{2, \rm z}\sigma_2^{2, \z} \cdots \sigma_j^{2, \rm z}$, and of course $\sigma^{1,0}_j\equiv \sigma^{2,0}_j\equiv\one$,
the fermionic Hubbard model transforms into a spin ladder, 
\begin{align}
&H^{\rm Hub} = 2  \sum_{j=1}^{n-1} \sum_{d=1}^2 \left ( \sigma^{d,+}_j \sigma^{d,-}_{j+1} +  \sigma^{d,-}_{j} \sigma^{d,+}_{j+1} \right ) + \frac{u}{2}   \sum_{j=1}^{n} \sigma^{1, \rm{z}}_j\sigma^{2, \rm{z}}_j \nonumber \\
&+ \frac{u}{2} \sigma^{1,\z}_{1} \sigma^{2,\z}_{1} + \frac{\mu_{\rm L}}{2}\left(\sigma^{1,\z}_{1} + \sigma^{2,\z}_{1}\right) \nonumber \\ 
&+\frac{u}{2} \sigma^{1,\z}_{n} \sigma^{2,\z}_{n} + \frac{\mu_{\rm R}}{2}\left(\sigma^{1,\z}_{n} + \sigma^{2,\z}_{n}\right), \label{Hhub} \\ 
&h^{\rm Hub} _{j,j+1}=2 \sum_{d=1}^2 \left ( \sigma^{d,+}_j \sigma^{d,-}_{j+1} +  \sigma^{d,-}_{j} \sigma^{d,+}_{j+1} \right ) \nonumber \\
&+ \frac{u}{2} \left (  \sigma^{1, \rm{z}}_j\sigma^{2, \rm{z}}_j+\sigma^{1, \rm{z}}_{j+1}\sigma^{2, \rm{z}}_{j+1} \right), 
\label{Hhubloc1}
\end{align}
where the superscripts $1,2$ denote on which leg of the ladder the spin operator acts. 

The Lindblad jump operators following the Wigner-Jordan transformation become,
\begin{align}
&L_1 = \sqrt{\varepsilon_{\rm L}}\sigma^{1,+}_1,\quad L_2 = \sqrt{\varepsilon_{\rm L}} \sigma^{2,+}_1,\nonumber \\ 
&L_3 = \sqrt{\varepsilon_{\rm R}} \sigma^{1,-}_n ,\quad L_4 = \sqrt{\varepsilon_{\rm R}}  \sigma^{2,-}_n .
\label{lindhubspin}
\end{align}
In Ref. \cite{PopkovProsen} a novel Lax operator (unrelated to the standard one found by Shastry \cite{Shastry}) was found which allows for the solution to the equation for the NESS \eqref{NESS}. The solution relies on a generalized \emph{Sutherland-Shastry} equation. We will first write the needed relations formally and then define the representations of the operators used. 

Let us define the operators $\bb{S},\acute{\bb{S}},\grave{\bb{S}},\bb{T},\acute{\bb{T}},\grave{\bb{T}} \in {\rm End}({\cal H}_{ a}\otimes {\cal H}_{ p})$ (i.e., acting over both, the auxiliary and the physical space), and $\bb{X},\bb{Y}\in {\rm End}({\cal H}_{a})$ (scalars over the physical space). We will also likewise define two local \emph{kinetic energy} densities for the two species of spins,
\begin{align}
&h^{\sigma^{1}}_{j,j+1}:=2 \sigma^{1,+}_j \sigma^{1,-}_{j+1} + 2 \sigma^{1,-}_{j} \sigma^{1,+}_{j+1},  \\
&h^{\sigma^2}_{j,j+1}:=2 \sigma^{2,+}_{j} \sigma^-_{2,j+1} + 2 \sigma^{2, -}_j \sigma^{2,+}_{j+1}.
\label{Hhubloc2}
\end{align}
We now assume that these operators have the following properties,
\begin{eqnarray}
&& [h^{\sigma^1}_{j,j+1},\mm{S}_j \mm{X} \mm{S}_{j+1}]  = \acute{\mm{S}}_j \mm{X} \mm{S}_{j+1} - \mm{S}_j\mm{X} \grave{\mm{S}}_{j+1}, \label{id1}\\
&& [h^{\sigma^2}_{j,j+1},\mm{T}_j \mm{X} \mm{T}_{j+1}] = \acute{\mm{T}}_j \mm{X} \mm{T}_{j+1} - \mm{T}_j\mm{X}\grave{\mm{T}}_{j+1}, \label{id2}\\
&& \mm{S}\acute{\mm{T}} + \mm{T}\acute{\mm{S}} - \grave{\mm{S}}\mm{T} - \grave{\mm{T}}\mm{S} = [\mm{Y}-u \sigma^{1, \z} \sigma^{2,\z},\mm{S}\mm{T}], \label{id3}\\
&& [\mm{S},\mm{T}] = 0, \label{id4}\\
&&[\mm{X},\mm{Y}] = 0. \label{id5}
\end{eqnarray}
It can be shown by applying these identities that a Lax operator, appealingly factorized as, 
\begin{equation}
\mm{L} = \mm{S}\mm{T}\mm{X}, \label{HubbardLax}
\end{equation}
and another operator, playing the role of its 'derivative',
\begin{equation}
\widetilde{\mm{L}} = \half(\mm{S}\acute{\mm{T}} + \mm{T}\acute{\mm{S}} + \grave{\mm{S}}\mm{T} + \grave{\mm{T}}\mm{S} - \{\mm{Y},\mm{S}\mm{T}\})\mm{X}, \label{DerLax}
\end{equation}
satisfy a generalized Sharstry-Sutherland relation,
\begin{equation}
[h^{\rm Hub}_{j,j+1},\mm{L}_j \mm{L}_{j+1}] = (\widetilde{\mm{L}}_j + \mm{Y} \mm{L}_j)\mm{L}_{j+1} - \mm{L}_j(\widetilde{\mm{L}}_{j+1} + \mm{L}_{j+1}\mm{Y}).
\label{suthsha}
\end{equation}
In \cite{PopkovProsen} a representation of these operators was found and we will now state it. We label the basis in the auxiliary space as ${\cal V}=\{0^{+},\frac{1}{2}^{+},\frac{1}{2}^{-},1^{-},1^{+},\frac{3}{2}^{+},\frac{3}{2}^{-},2^{-},2^{+}\ldots\}$ so that ${\cal H}_{\rm a}={\rm lsp}\{\ket{v}; v\in{\cal V}\}$. The reason for this is that these operators can be represented as the so-called \emph{walking graph states}, as discussed in \cite{ProsenHubbard,PopkovProsen,Prosenreview}. 
Then $\mm{S}$ is given by the following matrix representations, 
\begin{eqnarray}
&&\mm{S}^+ =\sqrt{2} \sum_{k=0}^\infty \left(\ket{k^+}\bra{k\!+\!\half^+} + \ket{k\!+\!\half^-}\bra{k\!+\!1^-}\right),\label{eq:Sp} \\
&&\mm{S}^-  = \sqrt{2} \sum_{k=0}^\infty (-1)^k \left(\ket{k\!+\!\half^+}\bra{k^+} + \ket{k\!+\!1^-}\bra{k\!+\!\half^-}\right),\nonumber\\
&&\mm{S}^0 = \sum_{k=0}^\infty \bigl(\ket{2k^+}\bra{2k^+} + \ket{2k\!+\!\half^+}\bra{2k\!+\!\half^+} \nonumber \\
&&\qquad+ \ket{2k\!+\!1^-}\bra{2k\!+\!1^-}+ \ket{2k\!+\!\half^-}\bra{2k\!+\!\half^-}\bigr)\nonumber\\
&&\quad+ \phi\sum_{k=1}^\infty \left( \ket{2k\!-\!\half^+}\bra{2k\!-\!\half^+} + \ket{2k^-}\bra{2k^-}\right), \nonumber \\
&&\mm{S}^\z = \sum_{k=1}^\infty \bigl(\ket{2k\!-\!1^+}\bra{2k\!-\!1^+} + \ket{2k\!-\!\half^+}\bra{2k\!-\!\half^+} \nonumber \\
&&\qquad + \ket{2k^-}\bra{2k^-} + \ket{2k\!+\!\half^-}\bra{2k\!+\!\half^-}\bigr)  \nonumber\\
&&\quad + \phi\sum_{k=0}^\infty\left(\ket{2k\!+\!\half^+}\bra{2k\!+\!\half^+} + \ket{2k\!+\!1^-}\bra{2k\!+\!1^-}\right), \nonumber
\end{eqnarray}
where $\phi\in\CC$ is a free parameter. We define an operator $\mm{G}\in{\rm End}({\cal H}_{ a}\otimes {\cal H}_{ p})$ which interchanges spins of different ladders in the physical space, i.e., $\mm{G} \sigma^{1, s} \mm{G} = \sigma^{2,s}$, $\mm{G} \sigma^{2, s} \mm{G} = \sigma^{1,s}$ and in the auxiliary space operates as $\mm{G}\ket{k^\pm} := \ket{k^\pm}$, $\mm{G} \ket{k\!+\!\half^\pm} := \ket{k\!+\!\half^\mp}$, $k \in \ZZ^+$, $\mm{T}$ is given by $\mm{T}=\mm{G} \mm{S} \mm{G}$. 

The operator $\mm{X}$ is given by the block diagonal matrix,
\begin{eqnarray}
\mm{X} &=& \ket{0^+}\bra{0^+} +
 \sum_{k=1}^\infty (-1)^{k}\!\!\!\sum_{\nu,\nu'\in\{-,+\}}   \ket{k^{\nu}} X^{\nu,\nu'}_k  \bra{k^{\nu'}}  \label{eq:Xansatz}\\
&+& w \sum_{k=0}^\infty (-1)^k \left(\ket{k\!+\!\half^+}\bra{k\!+\!\half^+} + \nonumber
\ket{k\!+\!\half^-}\bra{k\!+\!\half^-}\right), 
\end{eqnarray}
where $X_k=\{ X^{\nu,\nu'}_k \}_{\nu,\nu'\in\{-,+\}}$ are $2\times2$ matrices 
\begin{equation}
X_k(\phi,w)=\begin{pmatrix}
-(w+k u)w & 1-(w+k u)w(1-\phi^2)\cr
 -k u w & 1- k u w (1-\phi^2),
\end{pmatrix}.
\label{Xk}
\end{equation}
and $\phi,w \in \CC$ are free parameters. 
The operator $\mm{Y}$ is purely diagonal,
\begin{equation}
\mm{Y}=-2\phi u\sum_{k=0}^{\infty}\bigl(\ket{k^+}\bra{k^+} + \ket{k\!+\!1^-}\bra{k\!+\!1^-}\bigr).
\end{equation} 
Demanding that $\mm{X}$ is invertible (namely, $w\neq 0$, $\det X_k \neq 0$) we may write implicitly, 
\begin{eqnarray}
&&\acute{\mm{S}}^+\mm{X}  =  -2\sqrt{2}\sum_{k=1}^{\infty}(-1)^k X^{+-}_k \ket{k^-}\bra{k\!+\!\half^+}, \label{eq:Sb}\\
&&\acute{\mm{S}}^-\mm{X} =   -2\sqrt{2}\sum_{k=1}^{\infty} X^{-+}_k\ket{k^+}\bra{k\!-\!\half^-}, \nonumber\\
&&\mm{X}\grave{\mm{S}}^{+}  =  2\sqrt{2}\sum_{k=1}^{\infty}
(-1)^k X^{+-}_k\ket{k\!-\!\half^-}\bra{k^+} \nonumber\\
&&\mm{X}\grave{\mm{S}}^- =  -2\sqrt{2}\sum_{k=1}^{\infty}
X^{-+}_k \ket{k\!+\!\half^+}\bra{k^-}, \nonumber\\
&&\acute{\mm{S}}^0\mm{X}=\mm{X}\grave{\mm{S}}^0=
2\sum_{k=1}^{\infty}\bigl(w\ket{2k\!-\!1^+}\bra{2k\!-\!1^+}-w\ket{2k^-}\bra{2k^-} \nonumber \\
&&\quad - X^{++}_{2k-1}\ket{2k\!-\!\half^+}\bra{2k\!-\!\half^+}-X^{--}_{2k}\ket{2k\!-\!\half^-}\bra{2k\!-\!\half^-}\bigr) \nonumber \\
&&+2\phi
\sum_{k=0}^{\infty}
\bigl(-w \ket{2k^+}\bra{2k^+}+X^{--}_{2k+1}\ket{2k\!+\!\half^-}\bra{2k\!+\!\half^-}\bigr), \nonumber\\
&&\acute{\mm{S}}^\z\mm{X}=\mm{X}\grave{\mm{S}}^\z=
 2\sum_{k=0}^{\infty}\bigl(w \ket{2k\!+\!1^-}\bra{2k\!+\!1^-}-w\ket{2k^+}\bra{2k^+} \nonumber \\
 &&\quad + X^{++}_{2k}\ket{2k\!+\!\half^+}\bra{2k\!+\!\half^+} + X^{--}_{2k+1}\ket{2k\!+\!\half^-}\bra{2k\!+\!\half^-}\bigr) \nonumber \\
 && + 2\phi
\sum_{k=1}^{\infty}\bigl(w\ket{2k\!-\! 1^+}\bra{2k\!-\! 1^+}-X^{--}_{2k} \ket{2k\!-\!\half^-}\bra{2k\!-\!\half^-}\bigr),  \nonumber
\end{eqnarray}
and 
\begin{equation}
 \mm{G} \acute{\mm{S}} \mm{G} = \acute{\mm{T}},\quad
 \mm{G} \grave{\mm{S}} \mm{G} = \grave{\mm{T}}.\, \label{acutedef}
\end{equation}
It can then be directly shown that Eqs.~(\ref{id1}--\ref{id5}) are fulfilled and thus the Sutherland-Shastry equation \eqref{suthsha} holds. The solution to the equation for the NESS \eqref{NESS} proceeds in a similar fashion as in the XXZ case in Sect.~\ref{sec:XXZ}. The matrix product ansatz is given again in a factorized form as,
\begin{equation}
\rho_\infty = \frac{\Omega \Omega^\dagger K}{\tr \Omega \Omega^\dagger K}
\end{equation}
where $\Omega$ is again given by the highest weight state $\ket{0^+}$ expectation value of a monodromy matrix,
\begin{equation}
\Omega = \bra{0^+} \mm{L}_1(\phi,w)\mm{L}_2(\phi,w)\cdots  \mm{L}_n(\phi,w)\ket{0^+}
\label{SHub}
\end{equation}
and $K$ is a diagonal operator
\begin{equation}
K = K_1 K_2\ldots K_n,\quad K_n = \exp\left(\kappa (\sigma^{1,\z}_j+\sigma^{2,\z}_j)\right)
\end{equation}
with $\kappa = \half\log \Gamma_{\rm L}/\Gamma_{\rm R}$. The parameters $\phi,w$ are again determined with the boundary cancellation
mechanism and are given in terms of the system-bath coupling strengths as,
\begin{align}
\phi &= \frac{\Gamma_{\rm L}-\Gamma_{\rm R} - \ii (\mu_{\rm L}+\mu_{\rm R})}{\Gamma_{\rm L}+\Gamma_{\rm R} -
\ii (\mu_{\rm L}-\mu_{\rm R})}, \nonumber \\
w &=\frac{1}{4}\left(\mu_{\rm L}-\mu_{\rm R} + \ii\left(\Gamma_{\rm L}+\Gamma_{\rm R}\right)\right).
\label{Hubspectral}
\end{align}
It can be also shown that the fermion current scales asymptotically sub-diffusively as $~1/n^2$ for large $n$ \cite{PopkovProsen,ProsenHubbard,Prosenreview}. 

Curiously, the operator family $\Omega(\phi,w)$ which is generally non-diagonalizable again appears to be commuting $[\Omega(\phi,w),\Omega(\phi',w')] = 0$ and can be used to define novel conserved charges which break parity-hole symmetry of the Hubbard model (which is respected by all known local conserved charges).

\section{Stochastically boundary-driven bulk-deterministic classical systems}
\label{sec3}
Much like the previous section, in this section we will focus on the NESS of bulk deterministic system acted on by stochastic baths. The crucial difference is that, in this case, the bulk system will be classical. Remarkably, many similarities are present, likely reflecting a deeper connection between both classical and quantum integrable systems in this type of setup. We will begin by defining the model which we will later solve for both the NESS and a certain class of the decay modes (determining the finite-time dynamics of the system). In this section we present the results of \cite{ProsenMejia} and \cite{ProsenBuca}. 

\subsection{Deterministic classical cellular automata with stochastic boundaries}

The classical deterministic systems we will look at will be discrete space-time cellular automata. We focus on time-reversible cellular automata. The dynamics will occur on a diamond-shaped plaquette. Each site on the plaquette can take either a on (1) or off (0) value. The local dynamical rule is denoted $\chi : \ZZ_2\times \ZZ_2\times \ZZ_2 \to \ZZ_2$. Time runs in the north to south direction. This defines classical deterministic dynamics over a $1+1$ dimensional lattice $s_{x,t+1} = \chi(s_{x-1,t},s_{x,t-1},s_{x+1,t})$, $s_{x,t}\in \{0,1\}$, where only lattice sites $(x,t)$ of fixed, say {\em even} parity of $x\pm t$ are considered (the 1 dimensional lattice sites form a zig-zag pattern) where the size $n$ is assumed to be even. The micro-state of the system is thus given as a configuration along the saw patten, $\mathbf{s}=(s_1,s_2,\ldots,s_n) \equiv (s_{1,t+1},s_{2,t},s_{3,t+1},s_{4,t},\ldots,s_{n-1,t+1},s_{n,t})$.
The update procedure is a \emph{staggered} composition of updates of even sites 
\begin{equation}
s_{2 y,t+1}=\chi(s_{2y-1,t},s_{2y,t-1},s_{2y+1,t}),
\label{eeven}
\end{equation} and updates of odd sites 
\begin{equation}
s_{2y+1,t+2}=\chi(s_{2y,t+1},s_{2y+1,t},s_{2y+2,t+1}).
\label{eodd}
\end{equation} 
Further, we assume {\em reversibility}, i.e., $s_{\rm S} = \chi(s_{\rm W},s_{\rm N},s_{\rm E}) \Leftrightarrow s_{\rm N} = \chi(s_{\rm W},s_{\rm S},s_{\rm E})$ and
preservation of an empty state
$\chi(0,0,0)=0$.
This dynamics is thus a kind of discrete analog of second-order in time dynamics like Newton's laws of motion.

In this setup the time-evolution of the boundary sites $1,n$ is not determined. We introduce a stochastic update rule on the boundaries. In order to do that we first define a probability vector space ${\cal S} = (\RR^2)^{\otimes n}$ of probability distribution vectors $\mathbf{p}=(p_0,p_1,\ldots,p_{2^n-1}) \in {\cal S}$ of all the possible $2^n$ configurations (microstates) of the 1-dimensional pattern. 
The local rule 54 can then be given in terms of a three-site $2^3 \times 2^3$ permutation matrix $P$
\begin{equation}
P_{(s,s',s''),(t,t',t'')} = \delta_{s,t} \delta_{s',\chi(t,t',t'')} \delta_{s'',t''}, 
\end{equation}
where $\delta_{i,j}$ is the Kronecker delta. Note that this dynamics is deterministic. The local time evolution is embedded into ${\rm End}({\cal S})$ as the propagator $P_{k,k+1,k+2} = \one_{2^{k-1}}\otimes P \otimes \one_{2^{n-k-2}}$ acts on a triple of neighboring sites $k, k+1, k+2$. We introduce stochasticity on the boundaries through the following operators acting only on sites $1,2$ (the leftmost sites) and sites $n-1,n$ (the rightmost sites), 
\be
P^{\rm{L}}_{12}=P^{\rm{L}} \otimes \one_{2^{n-2}}, \qquad P^{\rm{R}}_{n-1,n}= \one_{2^{n-2}} \otimes P^{\rm{R}},
\ee
where $P^{\rm{L}}$ and $P^{\rm{R}}$ are stochastic matrices which map probability vectors to probability vectors (these matrices have non-negative elements and which in each column sum to 1 and thus conserve both the total probability and positivity of the probability vectors). The full time evolution of a probability vector is given via a Markov propagator $U$,
\begin{equation}
\mathbf{p}(t)=U^t \mathbf{p}(0),
\end{equation}
where $\mathbf{p}(0)$ is the initial probability distribution. The Mar\-kov propagator is split into two parts,
\be
U=U_{\rm o} U_{\rm e}. \label{U}
\ee
These are the two temporal layers which generate the previously discussed staggered dynamics (\ref{eeven},\ref{eodd}) for even and odd sites separately,
\begin{align}
&U_{\rm e}=P_{123}P_{345}\cdots P_{n-3,n-2,n-1} P^{\rm{R}}_{n-1,n}, \label{Ue}\\
&U_{\rm o}= P^{\rm{L}}_{12} P_{234} P_{456}  \cdots P_{n-2,n-1,n}. \label{Uo}
\end{align}
One may write the eigenvalue equation for the Markov propagator,
\begin{equation}
U \mathbf{p}=\Lambda \mathbf{p}, \label{eigeneq}
\end{equation}
which can be split into two coupled eigenvalue equations for the even and odd parts of the propagator,
\begin{equation}
U_{\rm e} \mathbf{p}=\Lambda_{\rm{L}} \mathbf{p'}, \qquad U_{\rm o} \mathbf{p'}=\Lambda_{\rm{R}} \mathbf{p}, \qquad \Lambda=\Lambda_{\rm{R}} \Lambda_{\rm{L}}. \label{eigensplit}
\end{equation}
For a graphical illustration of the above discussion see Fig.~\ref{fig:CA}.
\begin{figure}
 \centering	
\vspace{-1mm}
\includegraphics[width=0.8\columnwidth]{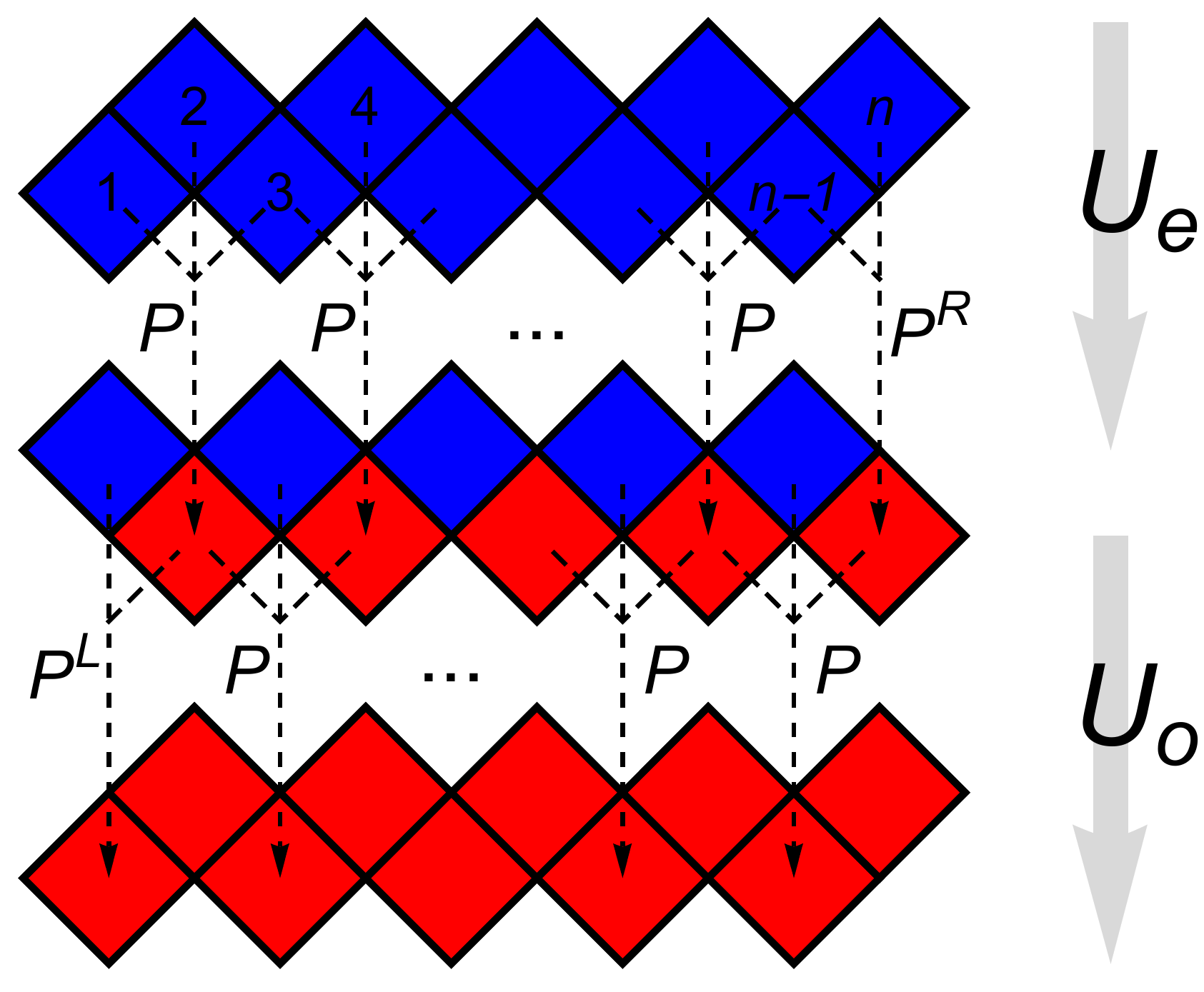}
\vspace{-1mm}
\caption{(From \cite{ProsenBuca}) A graphical explanation of the propagator \eqref{U}, composed of two split (time layers) $U_{\rm e}$ and $U_{\rm o}$ propagators, where each time layer is in turn composed of mutually commuting three site local permutation maps $P$ and two site boundary stochastic maps $P^{\rm L,R}$ (see Eqs.~(\ref{Ue},\ref{Uo})). Blue/red cells denotes the sites before/after the updates.}
\label{fig:CA}
\end{figure}

In terms of these eigenvectors the time-evolution can also be written as,
\be
\mathbf{p}(t) = \mathbf{p}_0 + \sum_{j\ge 1} c_j \Lambda_j^t \mathbf{p}_j, \label{timeeigvecs}
\ee 
where $c_j$ are constant given by the initial distribution $\mathbf{p}(0)$. We clearly see in Eq. \eqref{timeeigvecs} that eigenvector $\mathbf{p}_0$ is the NESS, as it does not decay in time. We also see that rest of the eigenvectors $\mathbf{p}_j$ ($j \ge 1$) decay exponentially with rate $\Lambda_j$. Any initial state $\mathbf{p}(0)$ will relax to NESS if all multiplicative rates are bounded away from unit circle $|\Lambda_j| <1$. 

A basic question is when is the NESS unique and approached from an arbitrary initial state, that is when is the time evolution ergodic and mixing. The answer to this question is provided by what we call \emph{holographic ergodicity} first studied for an integrable automaton in \cite{ProsenMejia}, but generally defined below for the first time for arbitrary `asymptotically abelian' bulk dynamics.

\subsection{Holographic ergodicity}
\label{sec:holo}
Before we explain what we mean by holographic ergodicity let us first state some concepts from linear algebra in relation to Perron-Frobenius theory \cite{matrixbook}. 
\begin{definition}[Irreducibility] A square matrix $U$ is call\-ed irreducible if for any pair configurations $\mathbf{s}, \mathbf{s}'$ one can find a natural number $t_0 \in \NN$ such that $(U^{t_0})_{\mathbf{s}, \mathbf{s}'}>0$. In other words, if $U$ connects every configuration with every other configuration after some finite number of time steps. \label{irr}
\end{definition}
\begin{definition}[Aperiodicity] A non-negative square matrix $U$ is called aperiodic if for some configuration $\mathbf{s}$ the greatest common divisor of repetition times $t \in\NN$ (for which $(U^{t_0})_{\mathbf{s}, \mathbf{s}}>0$, i.e., the same configuration repeats with finite probability) is 1. \label{ape} 
\end{definition}

It may be intuitive that these conditions are sufficient for the NESS to exist for any initial state and be unique: the first condition (irreducibility) ensures that the time evolution does not start and stay in one proper subspace, and the second condition (aperiodicity) ensures that the time evolution does not get stuck repeating periodically the same state instead of converging to one state (NESS).  

Indeed, if $U$ satisfies both of these conditions, then the Perron-Frobenius theorem \cite{matrixbook} guarantees that the NESS (eigenvector of $U$ with eigenvalue 1) is unique and all the other eigenvalues of $U$ lie inside the unit circle in the complex plane. 
A general discussion of the method of using this for classical boundary driven cellular automata follows. See the next section, Sect.~\ref{sec:rule54} for an example. 

\begin{theorem}[Holographic ergodicity] 
Suppose the following conditions hold. (i) The boundary stochastic processes are such that both states of the boundary cells are accessible
from any configuration
\begin{align}
P^{\rm L}_{(s,s'),(t,s')} > 0, 
\quad
P^{\rm R}_{(s',s),(s',t)} > 0,
\quad
\forall s,t,s'.
\end{align}
(ii) The bulk deterministic dynamics is asymptotically abe\-lian: considering the automaton defined on the infinite lattice $\ZZ$ and a pair of finite domains ${\cal A},{\cal B}\subset\ZZ$ such that the initial configuration is supported in ${\cal A}$, i.e. $s_{x,t}=0$, for all $x\not\in{\cal A},t\in\{0,1\}$, then $\exists t^*$, such that $s_{x,t}=0$, 
for all $x\in{\cal B}$, $t > t^*$.
Then, the matrix $U$ is aperiodic and irreducible, i.e. the corresponding Markov chain is ergodic and mixing.
\end{theorem}
\begin{proof}
Consider ${\cal A}={\cal B} = \{1,2,\ldots,n\}$.
Then for an arbitrary initial configuration $\mathbf{s}\in \ZZ_2^n$ defining initial-saw data $s_{x,t} = s_x$ for $x\in {\cal A}$, $t = x\, {\rm mod}\, 2$ and $s_{x,t}=0$ otherwise,
we have $t^*_{\mathbf{s}} \in \NN$, such that $s_{x,t} = 0$ for $1\le x \le n$ and $t \ge t^*_{\rm s}$. This means that a choice of boundary conditions
$\{s_{1,t}, t\; {\rm odd}\}$ for the left end and $\{s_{n,t}, t\; {\rm even}\}$ for the right end, due to property (i) always provide a walk in the Markov graph, i.e. a process with a finite probability which connects a microstate $\mathbf{s}$ with a vacuum $(0,0,\ldots,0)$ with a finite probability,
$(U^t)_{(0,0,\ldots,0),\mathbf{s}} > 0$, for $t \ge t^*_{\rm s}$.
Similarly we show that an arbitrary final configuration $\mathbf{s}'\in\ZZ_2^n$ can be created from a vacuum configuration in time $t^*(\mathbf{s}')$ by applying
time-reversal. Consequently, $(U^t)_{\mathbf{s}',\mathbf{s}} > 0$ for any $t \ge t^*_{\mathbf{s}}+t^*_{\mathbf{s}'}$, hence we have shown irreducibity and aperiodicity of $U$.
\end{proof}
 
For a diagramatic illustration of the proof see Fig.~\ref{IrredChart}. We note that our discussion of boundary driven reversible cellular automata so far straightforwardly generalize to more than 2-state cells.

\begin{figure}
 \centering	
\vspace{-1mm}
\includegraphics[width=0.47\columnwidth]{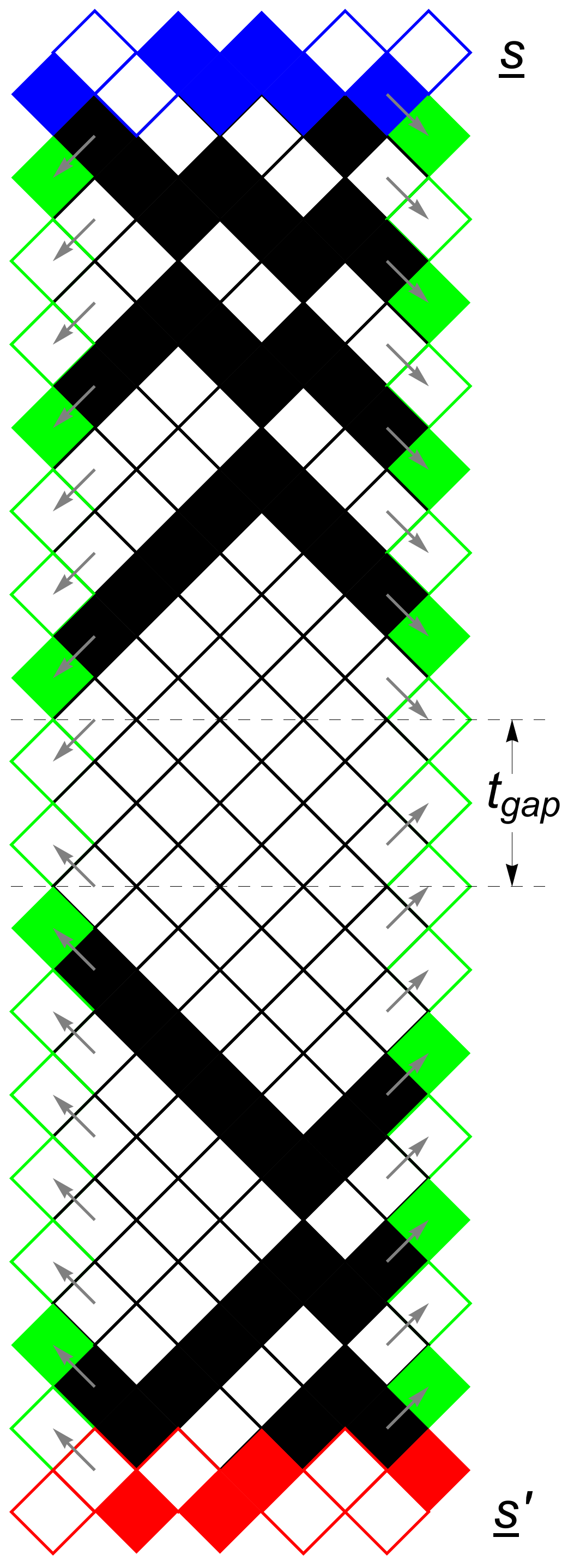}
\vspace{-1mm}
\caption{(From \cite{ProsenMejia}) Illustration of the proof of holographic ergodicity for the Rule 54 model. Two examples for configurations, 
$\mathbf{s}=(1,0,0,1,1,1,1,0,1,0)$ (blue) and $\mathbf{s}'=(0,0,1,0,1,1,0,0,0,1)$ (red), are connected with a walk for a generic value of probabilities 
$0<\alpha,\beta,\gamma,\delta < 1$ in at least $t_0=15$ time steps. The boundary cells are chosen and are shown as green cells. The values of the boundary cells are simply copies of the values of the cell in the direction indicated by grey arrows. Consequently $(U^{t_f+t_{gap}})_{\mathbf{s},\mathbf{s}'} > 0$ for any $t_{gap} \ge 0$ hence $U$ is irreducible and aperiodic (Defs.~\ref{irr},~\ref{ape}).}
\label{IrredChart}
\end{figure}

\subsection{The NESS and decay modes for the boundary driven Rule 54 cellular automaton}
\label{sec:rule54}
In this subsection we will focus on finding the solution for the Rule 54 automaton of Ref.~\cite{bobenko} with stochastic boundary driving \cite{ProsenMejia,ProsenBuca}. We will show how to construct the NESS in terms of a staggered matrix product ansatz and prove that it is unique. Later we will show how to deform this solution to obtain a class of decay modes of the model. The Rule 54 is an integrable cellular automaton with solitons traveling at the same velocity, $\pm 1$ \cite{bobenko}. The dynamics is quite simple -- the solitons travel either left or right and can scatter and incur a constant phase shift (see Ref.~\cite{bobenko}).

We begin by defining the model. A south site $s_{\rm S}$ is determined by a north, west and east sites according to the following rule (see Fig.~\ref{bobenkorule}),
\begin{align}
&s_{\rm S} = \chi({s_{\rm W},s_{\rm N},s_{\rm E}}) = s_{\rm N} + s_{\rm W} + s_{\rm E} +  s_{\rm W} s_{\rm E} \pmod{2}, \nonumber \\ 
&s_{\rm S},s_{\rm N},s_{\rm W},s_{\rm E} \in \ZZ_2.\label{chi}
\end{align}

\begin{figure}
 \centering	
\vspace{-1mm}
\includegraphics[width=\columnwidth]{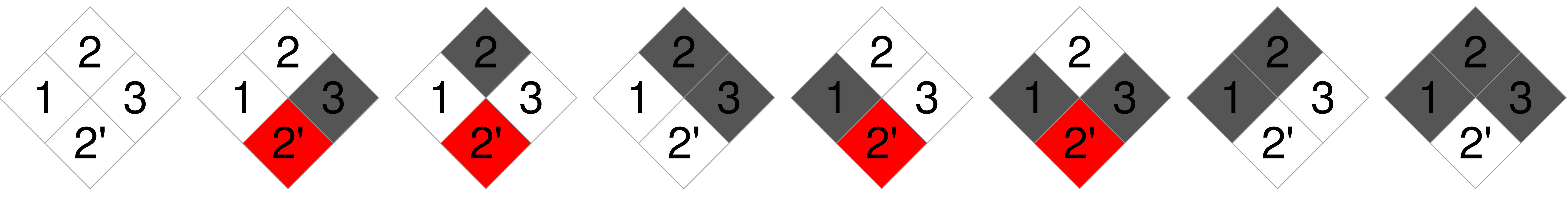}
\vspace{-1mm}
\caption{(From \cite{ProsenBuca}) The local rule 54. Each site can be either in state 0 or 1. State 1 is shown as dark grey at the current time step, red in the subsequent time step and 0 is always shown as white. The state of site 2 at the subsequent time step (labeled as 2') depends on the state of sites 1, 2, and 3 at the current time step. The full model is a combination of these plaquettes into a diamond-shaped lattice. It can be seen \cite{bobenko} that a combination of these plaquettes leads to solitons that can only scatter and incur a phase shift.}
\label{bobenkorule}
\end{figure}

This local time evolution rule can be encoded in the following permutation operator,
$$
P = \begin{pmatrix} 
\, 1\, & & & & & & & \cr
& & & \, 1\, & & & & \cr
& & \, 1\, & & & & & \cr
& \,1\, & & & & & & \cr
& & & & & & \,1\, & \cr
& & & & & & & \,1\, \cr
& & & & \,1\, & & & \cr
& & & & & \,1\, & & \end{pmatrix}.
$$
We take the following \emph{conditional} boundary driving\footnote{A solution with another type of driving is known \cite{ProsenMejia,ProsenBuca}, but it is conceptually similar to this one so for ease of presentation we omit discussing it here.},
\begin{align}
&P^{\rm{L}}=\left(
\begin{array}{cccc}
 \alpha  & 0 & \alpha  & 0 \\
 0 & \beta  & 0 & \beta  \\
 1-\alpha  & 0 & 1-\alpha  & 0 \\
 0 & 1-\beta  & 0 & 1-\beta  \\
\end{array}
\right),& \label{leftdriving} \\
&P^{\rm{R}}=\left(
\begin{array}{cccc}
 \gamma  & \gamma  & 0 & 0 \\
 1-\gamma  & 1-\gamma  & 0 & 0 \\
 0 & 0 & \delta  & \delta  \\
 0 & 0 & 1-\delta  & 1-\delta  \\
\end{array}
\right),&  
\label{rightdriving}
\end{align}
where $\alpha, \beta, \gamma, \delta \in [0,1]$ are driving rates parametrizing the left and the right bath. We call this driving conditional because in $P^{\rm{L}}_{12}$ ($P^{\rm{R}}_{n-1,n}$) the probability of changing the site 1 depends {\em only} on the state of the neighboring site 2 (changing the site $n$ depends only on the state of the site $n-1$).

We note that for an open set of boundary parameters, $0< \alpha,\beta,\gamma,\delta < 1$, the condition (i) of the holographic ergodicity theorem is satisfied, while the condition (ii) (asymptotic abelianess) is a simple consequence of the conserved solition currents \cite{bobenko} and the fact that solitions of Rule 54 
can not form bound states. See Fig.~\ref{IrredChart} for a graphic illustration of holographic ergodicity in the Rule 54 model.

\subsubsection{The matrix product ansatz for the NESS}

We will first endeavor to solve the equation for eigenvalue equation and find the NESS \eqref{eigeneq} with eigenvalue $\Lambda=1$ using a matrix product ansatz. To do this we will again introduce an auxiliary space ${\cal H}_a$. Let us first fix some notation. Vectors in the physical space will written in bold-face, as before for $\mathbf{p}(t)$. The number subscript of a physical space vector (or operator) denotes the site position in physical space on which it acts. When using component notation the components in physical space will be label with a binary `spin' index, i.e., $s\in\{0,1\}$. Matrices are denoted by capital roman letters. Row (column) vectors in the auxiliary space will be Dirac bras (kets). 

We begin by taking a \emph{staggered} matrix product ansatz for the split eigenvalue equation \eqref{eigensplit},
 \begin{align}
&\mathbf{p}=\bra{\mathbf{l}_1}\mmu{W}{2}\mmup{W}{3} \mmu{W}{4} \mmup{W}{5}\cdots\mmup{W}{n-3}\mmu{W}{n-2}\ket{\mathbf{r}_{n-1,n}}, \label{pness} \\
&\mathbf{p'}=\bra{\mathbf{l}'_{12}}\mmu{W}{3}\mmup{W}4\mmu{W}{5}\mmup{W}{6}\cdots\mmup{W}{n-2}\mmu{W}{n-1}\ket{\mathbf{r}'_{n}}. \label{ppness}
\end{align}
Assume that there exist operators $\mm{W}$, $\mm{W}'$ and $S$ satisfying the following algebraic relation,
\be
P_{123} \mathbf{W}_1 S\, \mathbf{W}_2 \mathbf{W}'_3=\mathbf{W}_1 \mathbf{W}'_2 \mathbf{W}_3 S.  \label{bulk1}
\ee
This algebraic relation is, to our knowledge, the fundamental equation encoding the integrability of the model. Since $P^2=\one$ from this follows also another equation obtained by multiplying Eq. \eqref{bulk1} from the left with $P_{123}$. In addition to this assume that $\mm{W}$ and $\mm{W}'$ can be interchanged and another dual bulk relation follows,
\be
P_{123} \mathbf{W}'_1 \mathbf{W}_2 \mathbf{W}'_3 S = \mathbf{W}'_1 S\, \mathbf{W}'_2 \mathbf{W}_3.  \label{bulk1b}
\ee
Further assume that there exists a representation of $\mm{W}$, $\mm{W}'$ and $S$ such that there as well exists a non-trivial solution to the following boundary conditions,
 \begin{align}
&P^{\rm{L}}_{12} \bra{\mathbf{l}'_{12}}=\lambda_{\rm{L}} \bra{\mathbf{l}_1} \mmu{W}{2} S, \label{bound1a} \\
&P_{123} \bra{\mathbf{l}_{1}} \mmu{W}{2} \mmup{W}{3}=\bra{\mathbf{l}'_{12}} \mmu{W}{3}S,\label{bound2a} \\
&P^{R}_{12} \ket{\mathbf{r}_{12}}=\mmup{W}{1}S\ket{\mathbf{r}'_2}, \label{bound3a} \\
&P_{123} \mmup{W}{1} \mmu{W}{2}\ket{\mathbf{r}'_{3}} =\lambda_{\rm{R}} \mmup{W}{1}S\ket{\mathbf{r}_{23}}. \label{bound4a}
\end{align}
Another necessary assumption will be $S^2=\one$. 

If all of these assumptions hold then $\mathbf{p}$ and $\mathbf{p}'$ fulfil the split eigenvalue equation with corresponding eigenvalues \eqref{eigensplit} (where we wrote here for sake of specificity as $\Lambda_{\rm{L,R}}=\lambda_{\rm{L,R}}$).

This can be seen by the following argument. First fully writing out $U_{\rm e} \mathbf{p}$ in terms of (\ref{Ue}) and the ansatz \eqref{pness}, one may use Eq.~\eqref{bound3a} in order to introduce the operator $S$ in a string 
$\cdots \mathbf{W}'_{n-3}\mathbf{W}_{n-2}\mathbf{W}'_{n-1}S$ (inside $\mathbf{p}$) and then one uses $\cdots P_{n-5,n-4,n-3}P_{n-3,n-2,n-1}$ via the dual bulk relation \eqref{bulk1b} in order to move the operator $S$ across to the left end. There it is then absorbed, via $S^2=\one$, by boundary equation \eqref{bound2a}. We thus finally get $U_{\rm e} \mathbf{p}=\lambda_{\rm R}\mathbf{p}'$.

Analogously, we proceed with $U_{\rm o} \mathbf{p}'$, in terms of (\ref{Uo}) and ansatz \eqref{ppness}, and now use the boundary Eqs.~(\ref{bound1a},\ref{bound4a}) to carry $S$
from left to right via the bulk relation \eqref{bulk1}, finishing with $U_{\rm e} \mathbf{p}'=\lambda_{\rm L}\mathbf{p}$.

We give a simple example of this cancellation mechanism for $n=6$,
\begin{align}
&U_{\rm o} U_{\rm e} \mathbf{p}=\nonumber \\
&U_{\rm o} P_{123}P_{345} P^{\rm{R}}_{5,6} \bra{\mathbf{l}_1}\mmu{W}{2}\mmup{W}{3} \mmu{W}{4} \ket{\mathbf{r}_{5,6}}= \nonumber \\
&U_{\rm o} P_{123}P_{345} \bra{\mathbf{l}_1}\mmu{W}{2}\mmup{W}{3} \mmu{W}{4} \mmup{W}{5}S\ket{\mathbf{r}'_6}=\nonumber \\
&U_{\rm o} P_{123} \bra{\mathbf{l}_1}\mmu{W}{2}\mathbf{W}'_3 S \mathbf{W}'_4 \mathbf{W}_5\ket{\mathbf{r}'_6}= \nonumber \\
&U_{\rm o} \bra{\mathbf{l}'_{12}} \mmu{W}{3}S^2 \mathbf{W}'_4 \mathbf{W}_5\ket{\mathbf{r}'_6}=U_{\rm o} \mathbf{p'}
\end{align}

There indeed exists a $4$-dimensional representation of $\mathbf{W}$, $\mathbf{W'}$ and an $S$ with the aforementioned properties. It was found by numerical experimentation \cite{ProsenBuca}. We write them out,
\be
W_0(\xi, \omega)=\left(
\begin{array}{cccc}
 1 & 1 & 0 & 0 \\
 0 & 0 & 0 & 0 \\
 \xi  & \xi  & 0 & 0 \\
 0 & 0 & 0 & 0 \\
\end{array}
\right), \qquad 
W_1(\xi, \omega)=\left(
\begin{array}{cccc}
 0 & 0 & 0 & 0 \\
 0 & 0 & \xi  & 1 \\
 0 & 0 & 0 & 0 \\
 0 & 0 & 1 & \omega  \\
\end{array}
\right),
\label{Wrepr}
\ee
noting that they depend on a pair of {\em spectral parameters}\footnote{The fact that we have free parameters in a representation is in general important for solving a matrix product ansatz, as this allows for non-trivial solutions to the boundary equations.} $\xi$ and $\omega$.
The operator $\mathbf{W}'$ is just $\mathbf{W}$ with the parameters $\xi$ and $\omega$ interchanged, i.e., (writing the dependence on the spectral parameters explicitly), 
\be
\mathbf{W}'(\xi, \omega)=\mathbf{W}(\omega,\xi),
\ee
and,
\be
S=\one \otimes \sigma^{\rm{x}}=\left(
\begin{array}{cccc}
 0 & 1 & 0 & 0 \\
 1 & 0 & 0 & 0 \\
 0 & 0 & 0 & 1 \\
 0 & 0 & 1 & 0 \\
\end{array}
\right).
\label{Sdef}
\ee
Using this representation we may solve the boundary equations for $\xi, \omega$ and find the split eigenvectors \eqref{pness}, \eqref{ppness} and eigenvalue equations \footnote{For explicit expressions of the boundary vectors $\ket{\mathbf{l}_1},\ket{\mathbf{r}_{n-1,n}}, \ldots$ the reader is referred to \cite{ProsenBuca}.} The eigenvalue equations are,
\begin{align}
&\frac{(\alpha +\beta -1)-\beta  \lambda_{\rm{L}}}{(\beta -1) \lambda_{L}^{2}}=\frac{\lambda_{\rm{R}} (\gamma -\lambda_{\rm{R}})}{(\delta -1) }, \\
&\frac{\lambda_{\rm{L}} (\alpha -\lambda_{\rm{L}})}{(\beta -1)}=\frac{ (\gamma +\delta -1)-\delta  \lambda_{\rm{R}}}{(\delta -1) \lambda_{\rm{R}}^2}.
\end{align}
Rewriting these equations in terms of the eigenvalue $\lambda=\lambda_{\rm{L}} \lambda_{\rm{R}}$ leads to
\begin{align}
&(\lambda-1)\times \label{char1}\\
&\Big(\lambda ^3+\lambda ^2 (1-\alpha  \gamma )+\lambda  [\beta  \delta -(\alpha +\beta -1) (\gamma +\delta -1)] \nonumber\\
&-(\alpha +\beta -1) (\gamma +\delta -1)\Big)=0. \nonumber
\end{align}
Clearly, $\lambda=1$ is always the solution and corresponds to the NESS. What remains is a cubic polynomial and we have three other solutions which correspond to three decay modes (whose eigenvalues do not depend on system size $n$). We call collectively these four eigenvalues the \emph{NESS orbital}.
Using this solution it can be shown that the current of both left- and right-moving solitons is ballistic in the NESS \cite{ProsenMejia}. 

We will now show how to deform this solution in order to obtain another orbital of solutions containing the second longest lived (apart from the NESS) Liouvillian eigenmode (called the \emph{leading decay mode}).

\subsection{The first (leading decay mode) orbital}
 
The physical intuition behind the generalization of the results from the previous section to another orbital is based on the spirit of the standard coordinate Bethe ansatz \cite{bethe1,bethe2,bethe3} generalized to the coordinate matrix product ansatz picture for the ASEP \cite{ASEPMatrix}. In this approach higher decay modes are understood as excitations upon the NESS which is understood as a vacuum. The leading decay mode is thus a superposition of one-particle excitations. 

The matrix product ansatz presented in this section was discovered through numerical experimentation in \cite{ProsenBuca}. It works as an exact analytic solution, but its generalization to all the decay modes, as well as a deeper understanding of its fundamental aspects, remain elusive. It will prove necessary to extend the $4$-dimensional from the previous section to an $8$-dimensional one. We therefore append notation from the previous subsection to include matrices in an extended 8-dimensional auxiliary space to be denoted with hats. We begin by generalizing the $\mathbf{W}$ operators to 8 dimensions,
\begin{align}
&\hat{W}_s= e_{11}\otimes W_s(\xi z, \omega/z)+ e_{22}\otimes W_s(\xi/ z, \omega z), \nonumber \\
& \hat{W}'_s=e_{11}\otimes W'_s(\xi z, \omega/z)+ e_{22}\otimes W'_s(\xi/ z, \omega z), \label{diagpart}
\end{align}
where $e_{i j} = \ket{i}\bra{j}$, $i,j\in \{1,2\}$ is the Weyl basis of $2\times 2$ matrices.  
We also introduce several new operators,
\begin{align}
&\hat{F}^{(k)}=\one_8+ e_{1 2} \otimes \frac{c_+ z^k F_+ +c_- z^{-k} F_-}{\xi \omega -1},  \nonumber \\
&\hat{F}'^{(k)}=\one_8+ e_{1 2} \otimes \frac{c_+ z^k F'_+ +c_- z^{-k} F'_-}{\xi \omega -1}, \nonumber \\
&\hat{G}^{(k)}=\one_8+ e_{1 2} \otimes \frac{c_+ z^k G_+ +c_- z^{-k} G_-}{\xi \omega -1},  \nonumber \\
& \hat{G}'^{(k)}=\one_8+ e_{1 2} \otimes \frac{c_+ z^k G'_+ +c_- z^{-k} G'_-}{\xi \omega -1}, \nonumber \\
&\hat{K}^{(k)}=\one_8+ e_{1 2} \otimes \frac{c_+ z^k K_+ +c_- z^{-k} K_-}{\xi \omega -1}, \nonumber\\
&\hat{L}^{(k)}=(z e_{11}+z^{-1} e_{22})\otimes \one_4+ e_{1 2} \otimes \frac{c_+ z^k L_+ +c_- z^{-k} L_-}{\xi \omega -1}, \label{FLhat} 
\end{align} 
which linearly depend on two amplitude parameters $c_+,c_-$. The parameter $z$ introduced here underlies the aforementioned physical intuition. These operators will essentially create left and right-moving excitation over the NESS with $z$ being related to the quasi-momentum of the quasi-particle excitation. 
We continue by deforming the bulk relations \eqref{bulk1} and \eqref{bulk1b},
\begin{align}
&P_{123}  \hat{K}^{(k-1)}\mathbf{\hat{W}}_1 \hat{S} \hat{G}^{(k)}\mathbf{\hat{W}}_2 \hat{F}^{(k+1)}\mathbf{\hat{W}}'_3\nonumber \\
&\;\;=\hat{F}'^{(k-1)}\mathbf{\hat{W}}_1 \hat{G}'^{(k)} \mathbf{\hat{W}}'_2  \hat{K}^{(k+1)}\mathbf{\hat{W}}_3 \hat{S}, \nonumber \\
&P_{123} \hat{G}'^{(k-1)}\mathbf{\hat{W}}'_1 \hat{F}'^{(k)} \mathbf{\hat{W}}_2  \hat{L}^{(k+1)}\mathbf{\hat{W}}'_3 \hat{S}\nonumber \\
&\;\;=\hat{L}^{(k-1)}\mathbf{\hat{W}}'_1 \hat{S} \hat{F}^{(k)}\mathbf{\hat{W}}'_2 \hat{G}^{(k+1)}\mathbf{\hat{W}}_3, \label{bulk2}
\end{align}
where $\hat{S}=\one_2 \otimes S$. 
Provided that a representation of $F_+, F_-,F'_+, \ldots$ exists such that there exists a non-trivial solution to the following boundary equations,
\begin{align}
&P^{\rm{L}}_{12} \bra{\mathbf{\hat{l}}'_{12}}=\Lambda_{\rm{L}} \bra{\mathbf{\hat{l}}_1} \hmp{G}0 \hmup{W}{2}, \label{bound1} \\
&P_{123} \bra{\mathbf{\hat{l}}_{1}} \hm{L}{0}\hmup{W}{2}\hat{S} \hm{F}{1}\hmup{W}{3}=\bra{\mathbf{\hat{l}}'_{12}} \hm{K}{1}\hmu{W}{3}\hat{S},\label{bound2} \\
&P^{R}_{12} \ket{\mathbf{\hat{r}}_{12}}=\hm{F}{n-3}\hmup{W}{1}\hat{S}\ket{\mathbf{\hat{r}}'_2}, \label{bound3} \\
&P_{123} \hm{G}{n-4}\hmup{W}{1} \hm{K}{n-3}\hmu{W}{2}\hat{S} \ket{\mathbf{\hat{r}}'_{3}} =\Lambda_{\rm{R}} \hm{L}{n-4}\hmup{W}{1}\hat{S}\ket{\mathbf{\hat{r}}_{23}}. \label{bound4}
\end{align}
Then, the following inhomogeneous (site-dependent) MPA
\begin{align}
 \mathbf{p}=&\bra{\mathbf{\hat{l}}_1}\hat{L}^{(0)}\WW'_2\hat{S}\hat{F}^{(1)}\mathbf{\WW}'_3\hat{G}^{(2)}\mathbf{\WW}_4 \cdots \nonumber \\
&\hat{F}^{(n-5)}\mathbf{\WW}'_{n-3}\hat{G}^{(n-4)}\mathbf{\WW}_{n-2}\ket{\mathbf{\hat{r}}_{n-1,n}}, \label{p} \\
  \mathbf{p'}=&\bra{\mathbf{\hat{l}}'_{12}}\hat{F}^{'(1)}\mathbf{\WW}_3\hat{G}^{'(2)}\mathbf{\WW}'_4\hat{F}^{'(3)}\mathbf{\WW}_5 \cdots \nonumber \\
&\hat{G}^{'(n-4)}\mathbf{\WW}'_{n-2}\hat{K}^{(n-3)}\mathbf{\WW}_{n-1}\hat{S}\ket{\mathbf{\hat{r}}'_{n}}, \label{pp}
\end{align}
gives an eigenvector of $U=U_{\rm e}U_{\rm o}$ (\ref{Ue},\ref{Uo}) with an eigenvalue $\Lambda=\Lambda_{\rm L}\Lambda_{\rm R}$.

This can be seen in the following way. First we note that, 
\be
[P^{\rm L}_{12},P_{234}]=0,\qquad
[P_{123},P^{\rm R}_{34}]=0.
 \label{boundcomm}
\ee
This commutativity condition allows us to use the following argument. First $U_{\rm e}$ acts on $\mathbf{p}$. The first operator in $U_{\rm e}$ acting is $P_{123}$ on the left boundary vector, which via \eqref{bound2} creates $\hm{K}{1}\hmu{W}{3}\hat{S}$. The subsequent $P$'s in $U_{\rm e}$ move $\hm{K}{1}\hmu{W}{3}\hat{S}$ to the right using the bulk algebra \eqref{bulk2}. Before the final $P_{n-3,n-2,n-1}$ acts, $P_{\rm{R}}$ acts (as it commutes with $P_{n-3,n-2,n-1}$ \eqref{boundcomm}) and it creates the operator $\hm{F}{n-3}\hmup{W}{1}\hat{S}$. This operator is needed for the final $P_{n-3,n-2,n-1}$ to move $\hm{K}{n-5}\hmu{W}{n-3}\hat{S}$ to the actual end as $\hm{K}{n-3}\hmu{W}{n-1}\hat{S}$, and thus finally creating $\mathbf{p'}$ \eqref{pp}. The action of the odd-part of the propagator $U_{\rm o}$ is analogous in reverse.

We give an example for $n=8$ and for the even part of the propagator $U_{\rm e}$ to illustrate this mechanism,
\begin{align}
&U_{\rm o}U_{\rm e} \mathbf{p}=\nonumber \\
&U_{\rm o} P_{123}P_{345}P_{567} P^{\rm{R}}_{7,8} \cdot \nonumber \\
&\bra{\mathbf{\hat{l}}_1}\hat{L}^{(0)}\WW'_2\hat{S}\hat{F}^{(1)}\mathbf{\WW}'_3\hat{G}^{(2)}\mathbf{\WW}_4 \hat{F}^{(3)}\mathbf{\WW}'_{5}\hat{G}^{(4)}\mathbf{\WW}_{6}\ket{\mathbf{\hat{r}}_{7,8}}\nonumber=\\
&U_{\rm o}P_{345}P_{567} P^{\rm{R}}_{7,8}\cdot \nonumber \\
&\bra{\mathbf{\hat{l}}'_{12}} \hm{K}{1}\hmu{W}{3}\hat{S}\hat{G}^{(2)}\mathbf{\WW}_4 \hat{F}^{(3)}\mathbf{\WW}'_{5}\hat{G}^{(4)}\mathbf{\WW}_{6}\ket{\mathbf{\hat{r}}_{7,8}}\nonumber=\\
&U_{\rm o}P_{567} P^{\rm{R}}_{7,8}\bra{\mathbf{\hat{l}}'_{12}} \hat{F}'^{(1)}\mathbf{\hat{W}}_1 \hat{G}'^{(2)} \mathbf{\hat{W}}'_2  \hat{K}^{(3)}\mathbf{\hat{W}}_3 \hat{S}\hat{G}^{(4)}\mathbf{\WW}_{6}\ket{\mathbf{\hat{r}}_{7,8}}\nonumber=\\
&U_{\rm o}P_{567} \cdot \nonumber \\
&\bra{\mathbf{\hat{l}}'_{12}} \hat{F}'^{(1)}\mathbf{\hat{W}}_1 \hat{G}'^{(2)} \mathbf{\hat{W}}'_2  \hat{K}^{(3)}\mathbf{\hat{W}}_3 \hat{S}\hat{G}^{(4)}\mathbf{\WW}_{6}\hm{F}{5}\hmup{W}{7}\hat{S}\ket{\mathbf{\hat{r}}'_8}\nonumber=\\
&U_{\rm o} \bra{\mathbf{\hat{l}}'_{12}} \hat{F}'^{(1)}\mathbf{\hat{W}}_1 \hat{G}'^{(2)} \mathbf{\hat{W}}'_2 \hat{F}'^{(3)}\mathbf{\hat{W}}_1 \hat{G}'^{(4)} \mathbf{\hat{W}}'_2  \hat{K}^{(5)}\mathbf{\hat{W}}_3 \hat{S}^2\ket{\mathbf{\hat{r}}'_8}= \nonumber \\
&U_{\rm o}  \mathbf{p'}.
\end{align}
Indeed, as before, a non-trivial representation has been found,
\begin{align}
&F_{+}=\left(
\begin{array}{cccc}
 0 & 0 & 0 & 0 \\
 0 & 0 & 0 & z \\
 0 & 0 & \frac{\xi  \omega -1}{\omega  z^2} & 0 \\
 0 & 0 & 0 & \xi  z^2 \\
\end{array} 
\right), \\
&F_{-}=\left(
\begin{array}{cccc}
 0 & 0 & 0 & 0 \\
 0 & 0 & 0 & \frac{1}{\xi ^2 z^3} \\
 0 & 0 & \frac{\xi  \omega -1}{\xi  z^2} & 0 \\
 0 & 0 & 0 & \omega +\frac{1}{\xi}\left(\frac{1}{z^2}-1\right) \\
\end{array}
\right), \nonumber \\
&F'_{+}=\left(
\begin{array}{cccc}
 0 & 0 & 0 & 0 \\
 0 & 0 & 0 & \frac{z^3}{\omega ^2} \\
 0 & 0 & \frac{z^2 (\xi  \omega -1)}{\omega } & 0 \\
 0 & 0 & 0 & \xi +\frac{z^2-1}{\omega } \\
\end{array}
\right),  \nonumber\\
&F'_{-}= \left(
\begin{array}{cccc}
 0 & 0 & 0 & 0 \\
 0 & 0 & 0 & \frac{1}{z} \\
 0 & 0 & \frac{z^2 (\xi  \omega -1)}{\xi } & 0 \\
 0 & 0 & 0 & \frac{\omega }{z^2} \\
\end{array} \nonumber 
\right), \nonumber \\
&G_{+}=\left(
\begin{array}{cccc}
 0 & 0 & 0 & 0 \\
 0 & 0 & 0 & 0 \\
 \xi ^2 & 0 & 0 & 0 \\
 -\frac{2 \xi }{z} & -1 & 0 & 0 \\
\end{array}
\right), \nonumber \\ 
&G_{-}=\left(
\begin{array}{cccc}
 0 & 0 & 0 & 0 \\
 0 & 0 & 0 & 0 \\
 \xi  \omega  & 0 & 0 & 0 \\
 -\frac{\xi  \omega  z^2+1}{\xi  z^3} & -\frac{\omega }{\xi  z^2} & 0 & 0 \\
\end{array}
\right), \nonumber \\
&G'_{+}= \left(
\begin{array}{cccc}
 0 & 0 & 0 & 0 \\
 0 & 0 & 0 & 0 \\
 \xi  \omega  & 0 & 0 & 0 \\
 -z^3 \left(\frac{\xi }{z^2}+\frac{1}{\omega}\right) & -\frac{\xi  z^2}{\omega } & 0 & 0 \\
\end{array}
\right),  \nonumber \\ 
&G'_{-}=\left(
\begin{array}{cccc}
 0 & 0 & 0 & 0 \\
 0 & 0 & 0 & 0 \\
 \omega ^2 & 0 & 0 & 0 \\
 -2 \omega  z & -1 & 0 & 0 \\
\end{array}
\right),\nonumber \\
&K_{+}=\left(
\begin{array}{cccc}
 0 & 0 & 0 & 0 \\
 0 & \xi  z^2 & 0 & 0 \\
 0 & 0 & \frac{z^2 (\xi  \omega -1)}{\omega } & 0 \\
 0 & z & 0 & 0 \\
\end{array}
\right),  \nonumber \\
& K_{-}=\left(
\begin{array}{cccc}
 0 & 0 & 0 & 0 \\
 0 & \omega +\frac{1}{\xi}\left(\frac{1}{z^2}-1\right) & 0 & 0 \\
 0 & 0 & \frac{z^2 (\xi  \omega -1)}{\xi } & 0 \\
 0 & \frac{1}{\xi ^2 z^3} & 0 & 0 \\
\end{array}
\right), \nonumber\\
&L_{+}=\left(
\begin{array}{cccc}
 0 & 0 & 0 & 0 \\
 0 & \xi  & 0 & -z \\
 \frac{\xi  \omega }{z} & 0 & 0 & 0 \\
 0 & 0 & -\frac{\xi  \omega +z^2}{\omega ^2 z} & -\frac{z^2}{\omega } \\
\end{array}
\right), \nonumber \\
&L_{-}=\left(
\begin{array}{cccc}
 0 & 0 & 0 & 0 \\
 0 & \frac{1}{\xi  z^2} & 0 & -\frac{\omega }{\xi  z} \\
 0 & 0 & \frac{\omega }{z^2} & 0 \\
 -2 \omega  & 0 & 0 & -\omega  \\
\end{array}
\right).\nonumber
\end{align}
Note that by setting $z=1, c_+=c_-=0$, we recover the original bulk algebra (\ref{bulk1},\ref{bulk1b}). The boundary equations Eqs.~(\ref{bound1}- \ref{bound4}) now are solved for $\xi,\omega,z,c_+,c_-\in\CC$ and $k\in\ZZ$ to obtain the eigenvectors and eigenvalues in the first orbital. Specifically, for the left end,
\begin{align}
&\xi=\frac{z (\alpha +\beta -1)-\beta  \Lambda_{\rm{L}}}{(\beta -1) \Lambda_{L}^{2}}, \label{xi1}\\
&\omega =\frac{\Lambda_{\rm{L}} (\alpha  z-\Lambda_{\rm{L}})}{(\beta -1) z}. \label{paraleft1}\\
&\frac{c_-}{c_+}=\frac{\Lambda_{L}^{4}}{z^4 (\alpha +\beta -1)}, \label{paraleft2}
\end{align}
 and for the right end,
 \begin{align}
&\xi=\frac{\Lambda_{\rm{R}} (\gamma  z-\Lambda_{\rm{R}})}{(\delta -1) z}, \label{xi2} \\
&\omega =\frac{z (\gamma +\delta -1)-\delta  \Lambda_{\rm{R}}}{(\delta -1) \Lambda_{\rm{R}}^2},\label{om2} \\
&\frac{c_+}{c_-}=\frac{\Lambda_{\rm{R}}^{4} z^{4 m+2}}{\gamma +\delta -1}, \label{pararight} 
\end{align}
where $m=\frac{n}{2}-2$. Pairwise identifying expressions for $\xi$, $\omega$ and $c_+/c_-$ from the left and right boundary equations, we obtain a triple 
of coupled equations for the unknowns $\Lambda_{\rm L},\Lambda_{\rm R}$ and $z$.

For details of the solution procedure and the form of the boundary vectors we refer the reader to \cite{ProsenBuca}. Using this solution it can be shown that the eigenvalue of the leading decay modes scales with large $n$ as $1/n$,
\begin{align}
&\Lambda_1(0)= \label{ldmth}\\ 
&\frac{[1-(\alpha +\beta -1) (\gamma +\delta -1)] \log [(\alpha +\beta -1) (\gamma+\delta -1)]}{2 (\alpha +\beta -1) (\gamma +\delta -1)+\alpha  \gamma -\beta  \delta-2}, \nonumber
\end{align}
which is consistent with the aforementioned ballistic transport \cite{ProsenMejia}. 
Based on these results a conjecture for the higher orbital eigenvalues was obtained. Let $p$ denote the orbital level which runs from $1$ to $m=n/2-2$,
\begin{align}
&\frac{z (\alpha +\beta -1)-\beta  \Lambda_{\rm{L}}}{(\beta -1) \Lambda_{L}^{2}}=\frac{\Lambda_{\rm{R}} (\gamma  z-\Lambda_{\rm{R}})}{(\delta -1) z}, \nonumber \\
& \frac{\Lambda_{\rm{L}} (\alpha  z-\Lambda_{\rm{L}})}{(\beta -1) z}=\frac{z (\gamma +\delta -1)-\delta  \Lambda_{\rm{R}}}{(\delta -1) \Lambda_{\rm{R}}^2},\nonumber \\
&(\alpha +\beta -1)^p(\gamma +\delta -1)^p=(\Lambda_{\rm{L}} \Lambda_{\rm{R}})^{4 p} z^{4(m-p)+2}, \label{conj}
\end{align}
while $p=1$ corresponds to a rigorous leading orbital obtained by combining Eqs. (\ref{xi1}--\ref{pararight}).
The results for the eigenvalues as solutions of Eqs. (\ref{conj}) are presented in Fig.~\ref{fig:orbitals} and are shown to agree perfectly with numerical diagonalization of the Markov matrix $U$.
\begin{figure}
 \centering	
\vspace{-1mm}
\includegraphics[width=\columnwidth]{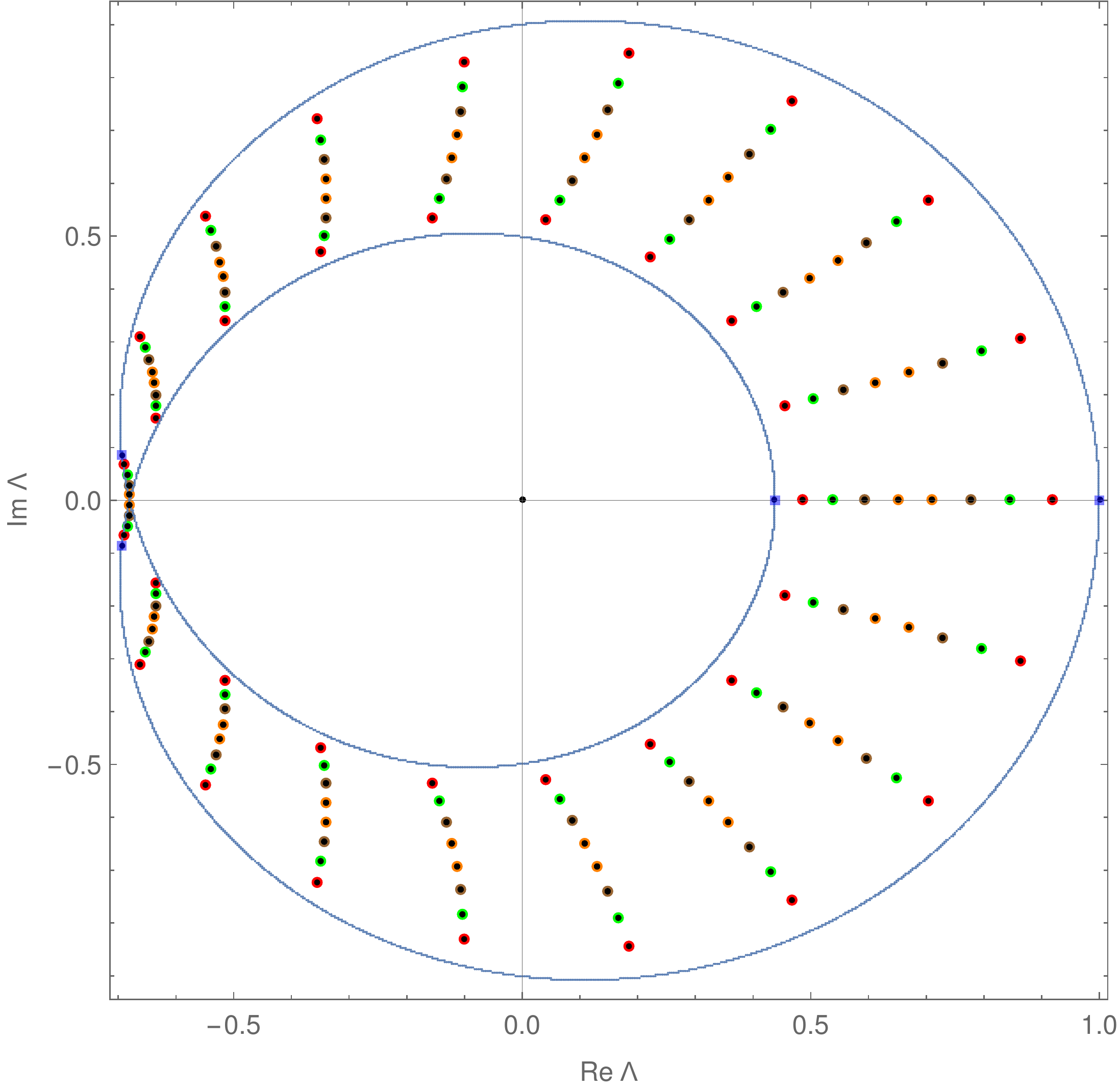}
\vspace{-1mm}
\caption{(taken from \cite{ProsenBuca}) The spectrum of the Markov propagator $U$ of the boundary driven Rule 54 for $n=12$ driving parameters $\alpha=1/4,\beta=1/3,\gamma=1/5,\delta=2/7$. The black dots show the numerical results. The red ($p=1$), green ($p=2$), brown ($p=3$), orange ($p=4$) points are solutions to the equations for the $p$-th orbital eigenvalues~(\ref{conj}). The blue squares are the roots of characteristic polynomial for the NESS-orbital. The blue curve is an algebraic curve to which the first orbital converges in the thermodynamic limit.}
\label{fig:orbitals}
\end{figure}

\section{Conclusion}

In this minireview we outlined recent results on \emph{interacting integrable} 1D quantum bulk coherent and \emph{interacting integrable} classical 1D bulk deterministic lattice systems driven by pairs of stochastic baths. The solution for the non-equilibrium steady state turned out to be, in all cases studied, given in a matrix product ansatz form. 

The quantum cases studied are the open maximally driven XXZ spin chain and the open maximally driven 1D fermionic Hubbard model. The classical model was a stochastically driven Rule 54 reversible cellular automaton. In the quantum case we managed to find the NESS using algebraic Bethe ansatz related methods and in the Rule 54 case the solution was found by means of a quite novel cancellation mechanism \eqref{bulk1}, which seems to be the fundamental integrability condition of the model. This algebraic cancellation relation allowed us to write the NESS in terms of a staggered matrix product ansatz. 

In the case of Rule 54 model we have also shown how to deform the solution for the NESS to find a class of higher decay modes of the Markov propagator in terms of an inhomogenous (site-dependent) matrix product ansatz. 

Many fundamentally interesting open questions remain. The analogies between the classical deterministic and quantum bulk coherent models seem numerous. One would like to understand whether there is a unifying mathematical framework for all solvable boundary stochastically driven cases. Perhaps such a framework would allow to find many other solvable models and the techniques to solve them.

Furthermore, one would like to go beyond the steady state by finding all the decay modes of the models, i.e. to fully diagonalize the Liouvillian propagator (Markov matrix). A class of the decay modes of the Rule 54 model (including the leading decay mode) were found using the aforementioned generalization of the bulk cancellation me\-chanism \eqref{bulk2}. One would like to find the rest of the decay modes, i.e. to provide a complete Bethe ansatz diagonalization of the Markov matrix.

An outstanding open question for several years now has been finding decay modes of the interacting boundary driven quantum models. This is still unsolved. 

Even more intriguing is the question on the possibility of genralization or classification of integrable dissipative boundaries. A pending issue is to understand why in quantum chains only the situations with pure source/pure sink boundaries appear to be exactly solvable. It seems that the concept of Sklyanin's reflection algebra is not useful in this case.

Finally, it would be very interesting to go beyond studying expectation values of the current and other interesting observables and compute the 
full counting statistics and the closely related large deviation theory in both cases, quantum and classical. 

\section*{Acknowledgement}

The work has been supported by European Research Council (ERC) Advanced grant 694544 -- OMNES and EPSRC programme grant EP/P009565/1.

\end{document}